\documentclass[12pt]{article}    

\usepackage[margin=0.75in]{geometry} 
\usepackage{amsmath,amssymb,amsthm}

\newtheorem{theorem}{Theorem}[section]
\newtheorem{proposition}[theorem]{Proposition}
\newtheorem{lemma}[theorem]{Lemma}

\newcommand{\be}{\begin{equation}}
\newcommand{\ee}{\end{equation}}

\newcommand{\bea}{\begin{eqnarray}}
\newcommand{\eea}{\end{eqnarray}}

\numberwithin{equation}{section}
\begin{document}

\title{Painlev\'{e} III$'$ and the Hankel Determinant Generated by a Singularly Perturbed Gaussian Weight}
\author{Chao Min\thanks{School of Mathematical Sciences, Huaqiao University, Quanzhou 362021, China; e-mail: chaomin@hqu.edu.cn}, Shulin Lyu\thanks{School of Mathematics (Zhuhai), Sun Yat-sen University, Zhuhai 519082, China; e-mail: lvshulin1989@163.com} and Yang Chen\thanks{Correspondence to: Yang Chen, Department of Mathematics, Faculty of Science and Technology, University of Macau, Macau, China; e-mail: yangbrookchen@yahoo.co.uk}}


\date{\today}
\maketitle
\begin{abstract}
In this paper, we study the Hankel determinant generated by a singularly perturbed Gaussian weight
$$
w(x,t)=\mathrm{e}^{-x^{2}-\frac{t}{x^{2}}},\;\;x\in(-\infty, \infty),\;\;t>0.
$$
By using the ladder operator approach associated with the orthogonal polynomials, we show that the logarithmic derivative of the Hankel determinant satisfies both a non-linear second order difference equation and a non-linear second order differential equation. The Hankel determinant also admits an integral representation involving a Painlev\'{e} III$'$. Furthermore, we consider the asymptotics of the Hankel determinant under a double scaling, i.e. $n\rightarrow\infty$ and $t\rightarrow 0$ such that $s=(2n+1)t$ is fixed. The asymptotic expansions of the scaled Hankel determinant for large $s$ and small $s$ are established, from which Dyson's constant appears.
\end{abstract}

$\mathbf{Keywords}$: Hankel determinant; Singularly perturbed Gaussian weight; Ladder operators;

Orthogonal polynomials; Painlev\'{e} III$'$; Double scaling.

$\mathbf{Mathematics\:\: Subject\:\: Classification\:\: 2010}$: 42C05, 33E17, 41A60.

\section{Introduction}
Hankel determinants, a fundamental object in Hermitian random matrix ensembles, determine the ``normalization constant'' $D_{n}$ of the multivariate probability density function \cite{Mehta}
$$
p(x_{1},x_{2},\ldots,x_{n})=\frac{1}{D_{n}}\prod_{1\leq j<k\leq n}\left(x_{k}-x_{j}\right)^{2}\prod_{j=1}^{n}w(x_{j}),
$$
where $\{x_j\}_{j=1}^{n}$ are the eigenvalues of any Hermitian matrix ensembles and $w(x)$ is a weight function. The object $D_{n}$ would depend on the parameters which appear in the weight.

Hankel determinants generated by the perturbed Gaussian, Laguerre and Jacobi weights have attracted a lot of interests in the last decade \cite{Basor2010,ChenDai,ChenIts,Dai,Han,Lyu,Xu,Zeng}.
These Hankel determinants and their asymptotics play an important role in the wireless communication system \cite{ChenHaqMckay2013,ChenMcKay2012}, and also many branches of applied mathematics and mathematical physics, such as enumeration problems \cite{BasorChenMekareeya,ChenJokela}, one-dimensional gas of impenetrable bosons \cite{Jimbo1980} and two-dimensional Ising model \cite{McCoy}.

Usually the study of these Hankel determinants is related to the Painlev\'{e} equations. For example, Chen and Its \cite{ChenIts} studied the Hankel determinant for a singularly perturbed Laguerre weight $x^{\alpha}\mathrm{e}^{-x-t/x}$ and showed that the Hankel determinant is the Jimbo-Miwa-Ueno isomonodromy $\tau$-function of a particular Painlev\'{e} III$'$, while Xu, Dai and Zhao \cite{Xu} obtained the uniform asymptotics of this Hankel determinant by using the Riemann-Hilbert approach based on the Deift-Zhou steepest descent method.  Chen and Dai \cite{ChenDai} showed that the logarithmic derivative of the Hankel determinant for the perturbed Jacobi weight $x^{\alpha}(1-x)^{\beta}\mathrm{e}^{-t/x}$ satisfies the Jimbo-Miwa-Okamoto $\sigma$-form of a particular Painlev\'{e} V.
Chen and Chen \cite{Chen2015,Chen2016} established the asymptotic expansion of these two Hankel determinants under a suitable double scaling. The Hankel determinant for a more general singularly perturbed weight $x^{\alpha}e^{-n(V(x)+(t/x)^k)}$ with $V(x)$ real analytic was studied in \cite{Atkin} and its aysmptotics under a double scaling was derived.

In this paper, we consider the Hankel determinant generated by a singularly perturbed Gaussian weight, namely,
$$
D_{n}(t):=\det\left(\int_{-\infty}^{\infty}x^{i+j}w(x,t)dx\right)_{i,j=0}^{n-1},
$$
where
$$
w(x,t):=\mathrm{e}^{-x^{2}-\frac{t}{x^{2}}},\;\;x\in(-\infty, \infty),\;\;t>0.
$$
With the multiplicative factor $\mathrm{e}^{-t/x^2}$, the weight $w(x,t)$ vanishes infinitely fast at $x=0$. This Hankel determinant is related to the Wigner time-delay distribution in chaotic cavities \cite{Texier}.

It is well known that the Hankel determinant $D_{n}(t)$ can be expressed as the product of the square of the $L^{2}$ norms of the monic polynomials orthogonal with respect to $w(x,t)$. In this paper we will use the ladder operators adapted to the orthogonal polynomials to deal with $D_n(t)$. This approach has been widely applied to the study of the Hankel determinants generated by the perturbed Gaussian, Laguerre and Jacobi weights mentioned above, also the problems in random matrix theory \cite{Basor2009,Basor2012,ChenZhang,Lyu2018,Min2018}.

We now introduce some information on the orthogonal polynomials. Let $P_{n}(x,t)$ be the monic polynomials of degree $n$ orthogonal with respect to the weight $w(x,t)$,
\be\label{or}
\int_{-\infty}^{\infty}P_{m}(x,t)P_{n}(x,t)w(x,t)dx=h_{n}(t)\delta_{mn},\;\;m, n=0,1,2,\ldots.
\ee
Since the weight $w(x,t)$ is even, $P_{n}(x,t)$ only contains the even or odd powers of $x$ for $n$ even or odd respectively \cite{Chihara}. That is,
\be\label{expan}
P_{n}(x,t)=x^{n}+\mathrm{p}(n,t)x^{n-2}+\cdots,
\ee
and we will see that $\mathrm{p}(n,t)$, the coefficient of $x^{n-2}$, is very important in the derivation of our main results.

For the orthogonal polynomials $P_{n}(x,t)$, we have the three-term recurrence relation \cite{Szego}
\be\label{rr}
xP_{n}(x,t)=P_{n+1}(x,t)+\beta_{n}(t)P_{n-1}(x,t)
\ee
with the initial conditions
$$
P_{0}(x,t)=1,\;\;\beta_{0}(t)P_{-1}(x,t)=0.
$$
An easy consequence of (\ref{or}), (\ref{expan}) and (\ref{rr}) gives
\be\label{be}
\beta_{n}(t)=\mathrm{p}(n,t)-\mathrm{p}(n+1,t)=\frac{h_{n}(t)}{h_{n-1}(t)},
\ee
and a telescopic sum yields
\be\label{sum}
\sum_{j=0}^{n-1}\beta_{j}(t)=-\mathrm{p}(n,t).
\ee
It is well known that \cite{Ismail}
\be\label{hankel}
D_{n}(t)=\prod_{j=0}^{n-1}h_{j}(t).
\ee
\noindent $\mathbf{Remark.}$ For the general weight, (\ref{rr}) should be replaced by
$$
xP_{n}(x)=P_{n+1}(x)+\alpha_{n}P_{n}(x)+\beta_{n}P_{n-1}(x).
$$
But in our case the weight is even, and this leads to $\alpha_{n}=0$.

The rest of this paper is organized as follows. In Sec. 2, we apply the ladder operators to our problem and obtain a non-linear second order difference equation satisfied by the log-derivative of $D_n(t)$. In Sec. 3, we show that the log-derivative of $D_n(t)$ also satisfies a non-linear second order differential equation. $D_n(t)$ can be expressed as an integral of an auxiliary quantity which satisfies a particular Painlev\'{e} III$'$. Sec. 4 is devoted to the derivation of the asymptotics of $D_{n}(t)$ under a double scaling, i.e. $n\rightarrow\infty, t\rightarrow 0$ such that $s=(2n+1)t$ is fixed. By using the results of Chen and Chen \cite{Chen2015}, we establish the large $s$ and small $s$ expansions. The conclusion is given in Sec. 5.

\section{Ladder Operators and Difference Equations}
In the following discussions, for convenience, we shall not display the $t$ dependence in $P_{n}(x)$, $w(x)$, $h_{n}$ and $\beta_{n}$ unless it is needed. The following lemmas can be found in \cite{ChenIts}; see also \cite{Basor2012,ChenIsmail} for reference.
\begin{lemma}
Suppose that $w(x)$ is a continuous even weight function defined on $(-\infty,\infty)$, and $w(-\infty)=w(\infty)=0$. The monic orthogonal polynomials with respect to $w(x)$ satisfy the following differential recurrence relations:
\be\label{lowering}
P_{n}'(z)=\beta_{n}A_{n}(z)P_{n-1}(z)-B_{n}(z)P_{n}(z),
\ee
\be\label{raising}
P_{n-1}'(z)=(B_{n}(z)+\mathrm{v}'(z))P_{n-1}(z)-A_{n-1}(z)P_{n}(z),
\ee
where
$$
A_{n}(z):=\frac{1}{h_{n}}\int_{-\infty}^{\infty}\frac{\mathrm{v}'(z)-\mathrm{v}'(y)}{z-y}P_{n}^{2}(y)w(y)dy,
$$
$$
B_{n}(z):=\frac{1}{h_{n-1}}\int_{-\infty}^{\infty}\frac{\mathrm{v}'(z)-\mathrm{v}'(y)}{z-y}P_{n}(y)P_{n-1}(y)w(y)dy,
$$
and $\mathrm{v}(z)=-\ln w(z)$.
\end{lemma}

\begin{lemma}\label{s1s2}
The functions $A_{n}(z)$ and $B_{n}(z)$ satisfy the conditions:
\be
B_{n+1}(z)+B_{n}(z)=z A_{n}(z)-\mathrm{v}'(z), \tag{$S_{1}$}
\ee
\be
1+z(B_{n+1}(z)-B_{n}(z))=\beta_{n+1}A_{n+1}(z)-\beta_{n}A_{n-1}(z). \tag{$S_{2}$}
\ee
\end{lemma}

The combination of ($S_{1}$) and ($S_{2}$) produces a sum rule.
\begin{lemma}\label{s2p}
$A_{n}(z)$, $B_{n}(z)$ and $\sum_{j=0}^{n-1}A_{j}(z)$ satisfy the identity
\be
B_{n}^{2}(z)+\mathrm{v}'(z)B_{n}(z)+\sum_{j=0}^{n-1}A_{j}(z)=\beta_{n}A_{n}(z)A_{n-1}(z). \tag{$S_{2}'$}
\ee
\end{lemma}
\noindent $\mathbf{Remark.}$ The three identities ($S_{1}$), ($S_{2}$) and ($S_{2}'$) are valid for $z\in \mathbb{C}\cup\{\infty\}$.

Solving for $P_{n-1}(z)$ from (\ref{lowering}) and plugging it into (\ref{raising}), we come to the following result with the aid of ($S_{2}'$).
\begin{lemma}
$P_{n}(z)$ satisfies the following second order differential equation:
$$
P_{n}''(z)-\left(\mathrm{v}'(z)+\frac{A_{n}'(z)}{A_{n}(z)}\right)P_{n}'(z)+\left(B_{n}'(z)-B_{n}(z)\frac{A_{n}'(z)}{A_{n}(z)}
+\sum_{j=0}^{n-1}A_{j}(z)\right)P_{n}(z)=0.
$$
\end{lemma}

For the problem at hand, we have
$$
w(x)=\mathrm{e}^{-x^{2}-\frac{t}{x^{2}}},\;\;x\in(-\infty, \infty),\;\;t>0,
$$
$$
\mathrm{v}(z)=z^{2}+\frac{t}{z^{2}},
$$
$$
\mathrm{v}'(z)=2z-\frac{2t}{z^3},
$$
$$
\frac{\mathrm{v}'(z)-\mathrm{v}'(y)}{z-y}=2+\frac{2t}{zy^3}+\frac{2t}{z^2y^2}+\frac{2t}{z^3y}.
$$
Using these results, we obtain the following proposition.
\begin{proposition}
For our problem, we have
\be\label{anz}
A_{n}(z)=2+\frac{R_{n}(t)}{z^2},
\ee
\be\label{bnz}
B_{n}(z)=\frac{r_{n}(t)}{z}+\frac{(1-(-1)^n)t}{z^3},
\ee
where
$$
R_{n}(t):=\frac{2t}{h_{n}}\int_{-\infty}^{\infty}\frac{1}{y^2} P_{n}^{2}(y)w(y)dy,
$$
$$
r_{n}(t):=\frac{2t}{h_{n-1}}\int_{-\infty}^{\infty}\frac{1}{y^3} P_{n}(y)P_{n-1}(y)w(y)dy.
$$
\end{proposition}

\begin{proof}
From the definition of $A_{n}(z)$, we find
\bea
A_{n}(z)&=&\frac{1}{h_{n}}\int_{-\infty}^{\infty}\left(2+\frac{2t}{zy^3}+\frac{2t}{z^2y^2}+\frac{2t}{z^3y}\right)P_{n}^{2}(y)w(y)dy\nonumber\\
&=&2+\frac{2t}{z^2h_{n}}\int_{-\infty}^{\infty}\frac{1}{y^2} P_{n}^{2}(y)w(y)dy\nonumber\\
&=&2+\frac{R_{n}(t)}{z^2},\nonumber
\eea
where we have used the parity of the integrand to get the result.\\
Similarly, we have
\bea
B_{n}(z)&=&\frac{1}{h_{n-1}}\int_{-\infty}^{\infty}\left(2+\frac{2t}{zy^3}+\frac{2t}{z^2y^2}+\frac{2t}{z^3y}\right)P_{n}(y)P_{n-1}(y)w(y)dy\nonumber\\
&=&\frac{2t}{zh_{n-1}}\int_{-\infty}^{\infty}\frac{1}{y^3} P_{n}(y)P_{n-1}(y)w(y)dy+\frac{2t}{z^3h_{n-1}}\int_{-\infty}^{\infty}\frac{1}{y} P_{n}(y)P_{n-1}(y)w(y)dy.\nonumber
\eea
From (\ref{or}) and (\ref{expan}), it follows that,
$$
\frac{1}{h_{n-1}}\int_{-\infty}^{\infty}\frac{1}{y} P_{n}(y)P_{n-1}(y)w(y)dy=\left\{
\begin{aligned}
&0,&n=0,2,4,\ldots,\\
&1,&n=1,3,5,\ldots.
\end{aligned}
\right.
$$
Hence, we find
\bea
B_{n}(z)
&=&\frac{2t}{zh_{n-1}}\int_{-\infty}^{\infty}\frac{1}{y^3} P_{n}(y)P_{n-1}(y)w(y)dy+\frac{(1-(-1)^n)t}{z^3}\nonumber\\
&=&\frac{r_{n}(t)}{z}+\frac{(1-(-1)^n)t}{z^3}.\nonumber
\eea
This completes the proof.
\end{proof}

Substituting (\ref{anz}) and (\ref{bnz}) into ($S_{1}$), we find
\be\label{s1}
R_{n}(t)=r_{n+1}(t)+r_{n}(t).
\ee
It follows that
\be\label{s11}
\sum_{j=0}^{n-1}R_{j}(t)=r_{n}(t)+2\sum_{j=0}^{n-1}r_{j}(t).
\ee
Similarly, from ($S_{2}'$), we obtain the following equalities:
\be\label{s21}
\beta_{n}=\frac{n+r_{n}(t)}{2},
\ee
\be\label{s22}
-2(-1)^n t\: r_{n}(t)=\beta_{n}R_{n}(t)R_{n-1}(t),
\ee
\be\label{s23}
r_{n}^2(t)+2(1-(-1)^n)t+\sum_{j=0}^{n-1}R_{j}(t)=2\beta_{n}R_{n-1}(t)+2\beta_{n}R_{n}(t).
\ee
By using (\ref{sum}), (\ref{s11}) and (\ref{s21}), we find
\be\label{s24}
4\mathrm{p}(n,t)=-4\sum_{j=0}^{n-1}\beta_{j}=-n(n-1)-2\sum_{j=0}^{n-1}r_{j}(t)=r_{n}(t)-\sum_{j=0}^{n-1}R_{j}(t)-n(n-1).
\ee
Eliminating $\sum_{j=0}^{n-1}R_{j}(t)$ from (\ref{s23}) and (\ref{s24}) gives
\be\label{s25}
4\mathrm{p}(n,t)=r_{n}(t)+r_{n}^2(t)+2(1-(-1)^n)t+\frac{4(-1)^n tr_{n}(t)}{R_{n}(t)}-(n+r_{n}(t))R_{n}(t)-n(n-1),
\ee
where we have made use of (\ref{s21}) and (\ref{s22}).

\noindent $\mathbf{Remark.}$ The equation (\ref{s23}) expresses the finite sum $\sum_{j=0}^{n-1}R_{j}(t)$, ultimately the $\sigma$-function as a rational function of $r_n(t)$ and $R_{n}(t)$. Without which one tends to find the third order differential equation, which is at first sight not the Painlev\'{e} equations.

By using \eqref{s1}, \eqref{s21}, \eqref{s22} and \eqref{s23}, we obtain the following results.
\begin{theorem}
The auxiliary quantities $r_{n}(t)$ and $R_{n}(t)$ satisfy the following non-linear second order difference equations:
\bea\label{rndiff}
-4(-1)^n t\: r_{n}(t)=\left(n+r_{n}(t)\right)\left(r_{n+1}(t)+r_{n}(t)\right)\left(r_{n}(t)+r_{n-1}(t)\right),
\eea
with the initial conditions
\[r_0(t)=0,\qquad r_1(t)=\frac{2t\int_{-\infty}^{\infty}\frac{1}{y^2}\mathrm{e}^{-y^2-\frac{t}{y^2}}dy}{\int_{-\infty}^{\infty}\mathrm{e}^{-y^2-\frac{t}{y^2}}dy},\]
and
\begin{equation}\label{Rndiff}
\begin{aligned}
&4(-1)^n t\left[(n+1)R_{n+1}(t)-nR_{n-1}(t)\right]+(2n+1)R_{n+1}(t)R_{n}(t)R_{n-1}(t)\\
&=\left[4(-1)^n t-R_{n+1}(t)R_{n}(t)\right]\left[4(-1)^n t+R_{n}(t)R_{n-1}(t)\right],
\end{aligned}
\end{equation}
with the initial conditions
\[R_0(t)=\frac{2t\int_{-\infty}^{\infty}\frac{1}{y^2}\mathrm{e}^{-y^2-\frac{t}{y^2}}dy}{\int_{-\infty}^{\infty}\mathrm{e}^{-y^2-\frac{t}{y^2}}dy},\qquad R_1(t)=\frac{2t\int_{-\infty}^{\infty}\mathrm{e}^{-y^2-\frac{t}{y^2}}dy}{\int_{-\infty}^{\infty}y^2\mathrm{e}^{-y^2-\frac{t}{y^2}}dy}.\]
\end{theorem}
\begin{proof}
Substituting (\ref{s1}) and (\ref{s21}) into (\ref{s22}), we obtain \eqref{rndiff}. To continue, we
insert \eqref{s21} into \eqref{s22} and get
$$
-4(-1)^n t\: r_{n}(t)=\left(n+r_{n}(t)\right)R_{n}(t)R_{n-1}(t).
$$
Solving for $r_n(t)$ from the above yields
\be\label{sol}
r_{n}(t)=-\frac{nR_{n}(t)R_{n-1}(t)}{4(-1)^n t+R_{n}(t)R_{n-1}(t)}.
\ee
Plugging it into (\ref{s1}) gives rise to \eqref{Rndiff}.
\end{proof}
Let
\be\label{sigma}
\sigma_{n}(t):=-\sum_{j=0}^{n-1}R_{j}(t).
\ee
We will find in the next section that $\sigma_{n}(t)$ is a quantity related to the Hankel determinant $D_{n}(t)$, namely,
$$
\sigma_{n}(t)=2t\frac{d}{dt}\ln D_{n}(t).
$$
\begin{theorem} The quantity $\sigma_n(t)$ satisfies a non-linear second order difference equation
\bea
&&4(-1)^n nt(\sigma_{n+1}(t)-\sigma_{n-1}(t))\left[4(-1)^n t+(\sigma_{n-1}(t)-\sigma_{n}(t))(\sigma_{n}(t)-\sigma_{n+1}(t))\right]\nonumber\\
&+&n^2(\sigma_{n-1}(t)-\sigma_{n}(t))^2(\sigma_{n}(t)-\sigma_{n+1}(t))^2\nonumber\\
&=&\left[\sigma_{n}(t)-2(1-(-1)^n)t\right]\left[4(-1)^n t+(\sigma_{n-1}(t)-\sigma_{n}(t))(\sigma_{n}(t)-\sigma_{n+1}(t))\right]^2\nonumber
\eea
with the initial conditions
\[\sigma_1(t)=-\frac{2t\int_{-\infty}^{\infty}\frac{1}{y^2}\mathrm{e}^{-y^2-\frac{t}{y^2}}dy}{\int_{-\infty}^{\infty}\mathrm{e}^{-y^2-\frac{t}{y^2}}dy},\qquad \sigma_2(t)=-\frac{2t\int_{-\infty}^{\infty}\frac{1}{y^2}\mathrm{e}^{-y^2-\frac{t}{y^2}}dy}{\int_{-\infty}^{\infty}\mathrm{e}^{-y^2-\frac{t}{y^2}}dy}
-\frac{2t\int_{-\infty}^{\infty}\mathrm{e}^{-y^2-\frac{t}{y^2}}dy}{\int_{-\infty}^{\infty}y^2\mathrm{e}^{-y^2-\frac{t}{y^2}}dy}.\]
\end{theorem}
\begin{proof}
From the definition of $\sigma_n(t)$, we easily see that
\be\label{rs}
R_{n}(t)=\sigma_{n}(t)-\sigma_{n+1}(t).
\ee
Substituting it into (\ref{sol}), we find
\be\label{rs1}
r_{n}(t)=-\frac{n(\sigma_{n}(t)-\sigma_{n+1}(t))(\sigma_{n-1}(t)-\sigma_{n}(t))}{4(-1)^n t+(\sigma_{n}(t)-\sigma_{n+1}(t))(\sigma_{n-1}(t)-\sigma_{n}(t))}.
\ee
By using (\ref{s21}) and (\ref{s22}), (\ref{s23}) becomes
\be\label{s231}
r_{n}^2(t)+2(1-(-1)^n)t-\sigma_{n}(t)=(n+r_{n}(t))R_{n}(t)-\frac{4(-1)^n t\: r_{n}(t)}{R_{n}(t)}.
\ee
Inserting (\ref{rs}) and (\ref{rs1}) into (\ref{s231}), we arrive at the desired difference equation.
\end{proof}

\section{Painlev\'{e} III$'$ and Its $\sigma$-form}
In this section, we take $t$-derivative of the orthogonality relations satisfied by $P_n(x,t)$. With the aid of the identities in the last section, we shall derive the second order ordinary differential equations satisfied by  $R_{n}(t)$, $r_{n}(t)$ and $\sigma_{n}(t)$, where $n$, the order of the Hankel matrix appears as a parameter.

We first take a derivative with respect to $t$ in the following equality
$$
\int_{-\infty}^{\infty}P_{n}^2(x,t)\mathrm{e}^{-x^{2}-\frac{t}{x^{2}}}dx=h_{n}(t),\;\;n=0,1,2,\ldots,
$$
and get
\be\label{eq1}
2t \frac{d}{dt}\ln h_{n}(t)=-R_{n}(t).
\ee
Hence, from (\ref{sigma}) and (\ref{hankel}), it follows that
$$
\sigma_{n}(t)=-\sum_{j=0}^{n-1}R_{j}(t)=2t\frac{d}{dt}\ln D_{n}(t).
$$
Using the fact that $\beta_{n}=\frac{h_{n}}{h_{n-1}}$, we have
$$
2t\frac{d}{dt}\ln\beta_{n}(t)=R_{n-1}(t)-R_{n}(t),
$$
that is,
\be\label{eq2}
2t\beta_{n}'(t)=\beta_{n}R_{n-1}(t)-\beta_{n}R_{n}(t).
\ee

On the other hand, differentiating the orthogonality relation
$$
\int_{-\infty}^{\infty}P_{n}(x,t)P_{n-2}(x,t)\mathrm{e}^{-x^{2}-\frac{t}{x^{2}}}dx=0,\;\;n=1,2,\ldots
$$
over $t$, we find
\be\label{pn}
\frac{d}{dt}\mathrm{p}(n,t)=\frac{1}{h_{n-2}}\int_{-\infty}^{\infty}\frac{1}{x^2}P_{n}(x,t)P_{n-2}(x,t)\mathrm{e}^{-x^{2}-\frac{t}{x^{2}}}dx.
\ee
Replacing $n$ by $n-1$ in (\ref{rr}), we have
\be\label{rr1}
\frac{P_{n-2}(x,t)}{h_{n-2}}=\frac{xP_{n-1}(x,t)}{h_{n-1}}-\frac{P_{n}(x,t)}{h_{n-1}}.
\ee
Substituting (\ref{rr1}) into (\ref{pn}), we obtain
\be\label{pnt}
2t\frac{d}{dt}\mathrm{p}(n,t)=(1-(-1)^n)t-\beta_{n}R_{n}(t).
\ee

Now we are ready to obtain the coupled differential equations satisfied by $r_n(t)$ and $R_n(t)$.
\begin{lemma} $r_n(t)$ satisfies a first order linear ordinary differential equation
\be\label{ricca1}
r_{n}'(t)=-\frac{2(-1)^n r_{n}(t)}{R_{n}(t)}-\frac{n+r_{n}(t)}{2t}R_{n}(t),
\ee
and $R_n(t)$ satisfies the Riccati equation
\be\label{ricca2}
2tR_{n}'(t)=R_{n}^2(t)+(1-2r_{n}(t))R_{n}(t)-4(-1)^n t.
\ee
\end{lemma}
\begin{proof}
According to (\ref{s21}) and (\ref{s22}), we first replace $\beta_{n}R_{n-1}(t)$ by $-2(-1)^n t\: r_{n}(t)/R_n(t)$ in \eqref{eq2} and then substitute $(n+r_n(t))/2$ for $\beta_n(t)$. Finally, we obtain \eqref{ricca1}.
To derive \eqref{ricca2}, we plug (\ref{s25}) into (\ref{pnt}) and replace $\beta_n(t)$ by $(n+r_n(t))/2$. By using (\ref{ricca1}) to eliminate $r_{n}'(t)$ in the resulting equation, we obtain the desired result.
\end{proof}
\noindent $\mathbf{Remark.}$ One finds in other problems, $r_n$ satisfied a Riccati equation \cite{Basor2009,Basor2010,ChenIts,Han,Min2018,Min}, whereas in this problem $r_n(t)$ satisfied a linear differential equation.

\begin{theorem}
$R_{n}(t)$ satisfies the following non-linear second order differential equation
\be\label{pain3}
R_{n}''(t)=\frac{(R_{n}'(t))^2}{R_{n}(t)}-\frac{R_{n}'(t)}{t}+\frac{(2n+1)R_{n}^2(t)}{4t^2}-\frac{(-1)^n}{t}+\frac{R_{n}^3(t)}{4t^2}-\frac{4}{R_{n}(t)},
\ee
which is a particular Painlev\'{e} III$'$, i.e. $P_{\mathrm{III'}}(2n+1,-4(-1)^n,1,-16)$, following the convention of \cite{Ohyama}. The non-linear second order differential equation for $r_{n}(t)$ reads
\bea\label{diff}
&&\Big[2 t^2 r_n'(t) r_n''(t)+2 t r_n(t) (r_n'(t))^2-8 (-1)^n t r_n(t) r_n'(t)+t (r_n'(t))^2-4 (-1)^n n t r_n'(t)\nonumber\\
&-&8 (-1)^n r_n^3(t)-8 (-1)^n n r_n^2(t)\Big]^2=t \Big[t (r_n'(t))^2-4 (-1)^n r_n^2(t)-4 (-1)^n n r_n(t)\Big]\nonumber\\
&\cdot&\Big[2 t r_n''(t)+2 r_n(t) r_n'(t)+r_n'(t)-8 (-1)^n r_n(t)-4 (-1)^n n\Big]^2.
\eea
\end{theorem}
\begin{proof}
Solving $r_{n}(t)$ from (\ref{ricca2}) and substituting it into (\ref{ricca1}) leads to (\ref{pain3}). On the other hand, solving $R_{n}(t)$ from (\ref{ricca1})
and substituting either solution into (\ref{ricca2}), we obtain (\ref{diff}).
\end{proof}
\noindent $\mathbf{Remark.}$ Equation (\ref{diff}) may be transformed into a Chazy type equation \cite{Chazy1909,Chazy1911,Cosgrove}.

\begin{theorem}
$\sigma_{n}(t)$ satisfies the following non-linear second order differential equation
\begin{small}
\bea\label{sig}
&&\bigg\{\Big[4\sigma_{n}(t)-4(1+(-1)^n)t \sigma_{n}'(t)+(-1)^n t (\sigma_{n}'(t))^2\Big]\Big[4\sigma_{n}(t)-8t \sigma_{n}'(t)+8(1-(-1)^n)t\Big]\nonumber\\
&-&(-1)^nt\Big[2(1-(-1)^n)-\sigma_{n}'(t)-2t\sigma_{n}''(t)\Big]^2\bigg\}^2=256n^2\Big[\sigma_{n}(t)-2t \sigma_{n}'(t)+2(1-(-1)^n)t\Big]^3.
\eea
\end{small}
Suppose that $n\rightarrow\infty, t\rightarrow 0$ such that $s:=n^2 t$ is fixed and the following limit
$$
\sigma(s):=\lim_{n\rightarrow\infty}\sigma_{n}\left(\frac{s}{n^2}\right)
$$
exists. Then $\sigma(s)$ satisfies the following second order differential equation
\be\label{p3p}
4s^2 (\sigma''(s))^2+4s\sigma'(s)\sigma''(s)+8s(\sigma'(s))^3-4\sigma(s)(\sigma'(s))^2+(\sigma'(s))^2=0.
\ee
\end{theorem}
\begin{proof}
From (\ref{s23}) and with the definition of $\sigma_n(t)$, i.e. $\sigma_{n}(t)=-\sum_{j=0}^{n-1}R_{j}(t)$ , we get
\[
2\beta_{n}R_{n-1}(t)+2\beta_{n}R_{n}(t)=r_{n}^2(t)+2(1-(-1)^n)t-\sigma_{n}(t).
\]
Noting that $\beta_n(t)=(n+r_n(t))/2$, it follows from \eqref{eq2} that
\[
2\beta_{n}R_{n-1}(t)-2\beta_{n}R_{n}(t)=2t r_{n}'(t).
\]
The sum and difference of the above two equations result in
\[
4\beta_{n}R_{n-1}(t)=r_{n}^2(t)+2(1-(-1)^n)t-\sigma_{n}(t)+2t r_{n}'(t)
\]
and
\be\label{dif}
4\beta_{n}R_{n}(t)=r_{n}^2(t)+2(1-(-1)^n)t-\sigma_{n}(t)-2t r_{n}'(t),
\ee
respectively. The product of them leads to
\be\label{prod}
-16(-1)^n(n+r_{n}(t))tr_{n}(t)=\left(r_{n}^2(t)+2(1-(-1)^n)t-\sigma_{n}(t)\right)^2-4t^2(r_{n}'(t))^2,
\ee
where we have made use of (\ref{s21}) and (\ref{s22}).

From (\ref{s24}) and with the definition of $\sigma_n(t)$, we find
\be\label{pnt1}
4\mathrm{p}(n,t)=\sigma_{n}(t)+r_{n}(t)-n(n-1).
\ee
Substituting (\ref{pnt1}) into (\ref{pnt}) yields
\be\label{equ}
2(1-(-1)^n)t-2\beta_{n}R_{n}(t)=t \sigma_{n}'(t)+t r_{n}'(t).
\ee
Combining (\ref{dif}) with (\ref{equ}) to eliminate $\beta_{n}R_{n}(t)$, we establish the relation between $r_{n}(t)$ and $\sigma_{n}(t)$:
\be\label{rn}
r_{n}^2(t)=\sigma_{n}(t)-2t \sigma_{n}'(t)+2(1-(-1)^n)t.
\ee
Taking a derivative on both sides of this equation gives
$$
r_{n}'(t)=\frac{2(1-(-1)^n)-\sigma_{n}'(t)-2t\sigma_{n}''(t)}{2r_{n}(t)},
$$
so that, by using \eqref{rn} again, we have
\bea
\left(r_{n}'(t)\right)^2&=&\frac{\left[2(1-(-1)^n)-\sigma_{n}'(t)-2t\sigma_{n}''(t)\right]^2}{4r_{n}^2(t)}\nonumber\\
&=&\frac{\left[2(1-(-1)^n)-\sigma_{n}'(t)-2t\sigma_{n}''(t)\right]^2}{4\left[\sigma_{n}(t)-2t \sigma_{n}'(t)+2(1-(-1)^n)t\right]}.\label{rnp}
\eea
Plugging (\ref{rn}) and (\ref{rnp}) into (\ref{prod}), we find a linear equation for $r_n(t)$ and the solution is given by
\[
r_{n}(t)=-\frac{4\sigma_{n}(t)-4(1+(-1)^n)t \sigma_{n}'(t)+(-1)^n t (\sigma_{n}'(t))^2}{4n}+\frac{(-1)^nt\left[2(1-(-1)^n)-\sigma_{n}'(t)-2t\sigma_{n}''(t)\right]^2}{16n\left[\sigma_{n}(t)-2t \sigma_{n}'(t)+2(1-(-1)^n)t\right]}.
\]
Inserting it into (\ref{rn}), we finally arrive at \eqref{sig}.

To continue, we suppose that $n\rightarrow\infty$ and $t\rightarrow 0$ such that $s:=n^2 t$ is fixed, and define
$$
\sigma(s):=\lim_{n\rightarrow\infty}\sigma_{n}\left(\frac{s}{n^2}\right).
$$
After the change of variables, equation (\ref{sig}) becomes
\begin{small}
\bea
&&s^2\Big[4s^2 (\sigma''(s))^2+4s\sigma'(s)\sigma''(s)+8s(\sigma'(s))^3-4\sigma(s)(\sigma'(s))^2+(\sigma'(s))^2\Big]^2\nonumber\\
&-&\frac{1}{n^2}\Big\{64 (1-(-1)^n) s^5 (\sigma ''(s))^3+(\sigma ''(s))^2 \big[320 s^5 (\sigma '(s))^2+192 (-1)^n s^5 (\sigma '(s))^2
-128 s^4 \sigma (s) \sigma '(s)\nonumber\\
&-&384 (-1)^n s^4 \sigma (s) \sigma '(s)+96(1- (-1)^n) s^4 \sigma '(s)+128 (-1)^n s^3 \sigma^2 (s)\big]\nonumber\\
&+&\sigma ''(s) \big[448 s^4 (\sigma '(s))^3+64 (-1)^n s^4 (\sigma '(s))^3-192 s^3 \sigma (s) (\sigma '(s))^2
-320 (-1)^n s^3 \sigma (s) (\sigma '(s))^2\nonumber\\
&+&48(1- (-1)^n) s^3 (\sigma '(s))^2+128 (-1)^n s^2 \sigma^2 (s) \sigma '(s)\big]\nonumber\\
&+&\left(640 s^4+384 (-1)^n s^4\right) (\sigma '(s))^5-\left[576 s^3 \sigma (s)
+960 (-1)^n s^3 \sigma (s)-144 s^3+16 (-1)^n s^3\right] (\sigma '(s))^4\nonumber\\
&+&\left[128 s^2 \sigma^2 (s)+640 (-1)^n s^2 \sigma^2 (s)-64 (1+(-1)^n) s^2 \sigma (s)-2048 s^3+8(1- (-1)^n) s^2\right] (\sigma '(s))^3\nonumber\\
&-&\left[128 (-1)^n s \sigma^3 (s)-32 (-1)^n s \sigma^2 (s)-3072 s^2 \sigma (s)\right] (\sigma '(s))^2-1536 s \sigma^2 (s) \sigma '(s)+256 \sigma^3 (s)\Big\}\nonumber\\
&+&O\left(\frac{1}{n^4}\right)=0.\nonumber
\eea
\end{small}
Letting $n\rightarrow\infty$ and keeping only the highest order term, we obtain \eqref{p3p}, following Min and Chen \cite{Min2018}.
\end{proof}
\noindent $\mathbf{Remark.}$
We call (\ref{p3p}) the $\sigma$ form of the Painlev\'{e} III$'$ since it is from the Painlev\'{e} III$'$ equation.

In the end of this section, we show the integral representation of $D_n(t)$ in terms of $R_n(t)$.
\begin{theorem}
$$
\ln\frac{ D_n(t)}{D_{n}(0)}=\int_{0}^{t}\left[\frac{1}{4}+2s-nR_{n}(s)-\frac{R_{n}^2(s)}{4}-\frac{sR_{n}'(s)}{R_{n}(s)}
+\frac{s^2((R_{n}'(s))^2-4)}{R_{n}^2(s)}\right]\frac{ds}{2s}
$$
\end{theorem}
\begin{proof}
substituting (\ref{s25}) into (\ref{pnt1}), and using (\ref{ricca2}) to eliminate $r_{n}(t)$, we obtain the expression of $\sigma_{n}(t)$ in terms of $R_{n}(t)$:
$$
\sigma_{n}(t)=\frac{1}{4}+2t-nR_{n}(t)-\frac{R_{n}^2(t)}{4}-\frac{tR_{n}'(t)}{R_{n}(t)}+\frac{t^2((R_{n}'(t))^2-4)}{R_{n}^2(t)}.
$$
The theorem is established by using the relation $\sigma_n(t)=2t\frac{d}{dt}\ln D_n(t)$.
\end{proof}

\section{Double Scaling Analysis}
In this section, we consider the large $n$ asymptotic behavior of $R_n(t), \sigma_n(t)$ and $D_n(t)$ under a double scaling. We connect our problem with the Hankel determinant generated by a singularly perturbed Laguerre weight which was studied in \cite{ChenIts,Chen2015}. Using their results, we are able to obtain the asymptotics of the scaled Hankel determinant, including the constant term.

Let $\tilde{D}_{n}(t,\alpha)$ be the Hankel determinant generated by the singularly perturbed Laguerre weight $x^{\alpha}\mathrm{e}^{-x-\frac{t}{x}}$, $x\in [0,\infty),\;\alpha>-1,\;t>0$, namely,
$$
\tilde{D}_{n}(t,\alpha):=\det\left(\int_{0}^{\infty}x^{i+j}x^{\alpha}\mathrm{e}^{-x-\frac{t}{x}}dx\right)_{i,j=0}^{n-1}.
$$
The orthogonality reads,
$$
\int_{0}^{\infty}\tilde{P}_{j}(x,\alpha)\tilde{P}_{k}(x,\alpha)x^{\alpha}\mathrm{e}^{-x-\frac{t}{x}}dx=\tilde{h}_{j}(t,\alpha)\delta_{jk},
$$
where $\tilde{P}_{j}(x,\alpha)$ are the monic polynomials of degree $j$ orthogonal with respect to the weight  $x^{\alpha}\mathrm{e}^{-x-\frac{t}{x}}$.

This Hankel determinant $\tilde{D}_{n}(t,\alpha)$ was studied by Chen and Its \cite{ChenIts}, and they obtained the following result.
\begin{lemma}
The quantity
\[H_{n}(t,\alpha):=t \frac{d}{dt}\ln \tilde{D}_{n}(t,\alpha)\]
satisfies the following non-linear second order differential equation,
\begin{small}
\be\label{p3}
(t H_{n}''(t,\alpha))^2=\left[n-(2n+\alpha)H_{n}'(t,\alpha)\right]^2-4\left[n(n+\alpha)+t H_{n}'(t,\alpha)-H_{n}(t,\alpha)\right]H_{n}'(t,\alpha)(H_{n}'(t,\alpha)-1).
\ee
\end{small}
\end{lemma}

\noindent $\mathbf{Remark.}$ The equation (\ref{p3}) can be transformed to the Jimbo-Miwa-Okamoto $\sigma$-form of the Painlev\'{e} III \cite{Jimbo1981}; see \cite{ChenIts} for the detailed explanation.

From the orthogonality condition (\ref{or}), we find for $n=0,1,2,\ldots,$
\bea
h_{2n}&=&\int_{-\infty}^{\infty}P_{2n}^2(x)\mathrm{e}^{-x^2-\frac{t}{x^2}}dx\nonumber\\
&=&2\int_{0}^{\infty}P_{2n}^2(x)\mathrm{e}^{-x^2-\frac{t}{x^2}}dx\nonumber\\
&=&\int_{0}^{\infty}P_{2n}^2(\sqrt{y})y^{-\frac{1}{2}}\mathrm{e}^{-y-\frac{t}{y}}dy\nonumber\\
&=&\int_{0}^{\infty}\tilde{P}_{n}^2\left(y,-\frac{1}{2}\right)y^{-\frac{1}{2}}\mathrm{e}^{-y-\frac{t}{y}}dy\nonumber\\
&=&\tilde{h}_{n}\left(t,-\frac{1}{2}\right),\nonumber
\eea
and
\bea
h_{2n+1}&=&\int_{-\infty}^{\infty}P_{2n+1}^2(x)\mathrm{e}^{-x^2-\frac{t}{x^2}}dx\nonumber\\
&=&2\int_{0}^{\infty}P_{2n+1}^2(x)\mathrm{e}^{-x^2-\frac{t}{x^2}}dx\nonumber\\
&=&\int_{0}^{\infty}P_{2n+1}^2(\sqrt{y})y^{-\frac{1}{2}}\mathrm{e}^{-y-\frac{t}{y}}dy\nonumber\\
&=&\int_{0}^{\infty}\left(y^{-\frac{1}{2}}P_{2n+1}(\sqrt{y})\right)^2y^{\frac{1}{2}}\mathrm{e}^{-y-\frac{t}{y}}dy\nonumber\\
&=&\int_{0}^{\infty}\tilde{P}_{n}^2\left(y,\frac{1}{2}\right)y^{\frac{1}{2}}\mathrm{e}^{-y-\frac{t}{y}}dy\nonumber\\
&=&\tilde{h}_{n}\left(t,\frac{1}{2}\right).\nonumber
\eea
Hence, we establish the relation between $D_{n}(t)$ and $\tilde{D}_{n}(t,\alpha)$, with $n=0,1,2,\ldots$,
\be\label{hd1}
D_{2n}(t)=\tilde{D}_{n}\left(t,\frac{1}{2}\right)\tilde{D}_{n}\left(t,-\frac{1}{2}\right),
\ee
\be\label{hd2}
D_{2n+1}(t)=\tilde{D}_{n}\left(t,\frac{1}{2}\right)\tilde{D}_{n+1}(t,-\frac{1}{2}).
\ee
Furthermore, since $\sigma_n(t)=2t\frac{d}{dt}\ln D_{n}(t)$ and $H_{n}(t,\alpha)=t \frac{d}{dt}\ln \tilde{D}_{n}(t,\alpha)$, we have for $n=0,1,2,\ldots,$
\be\label{re3}
\sigma_{2n}(t)=2\left[H_{n}\left(t,\frac{1}{2}\right)+H_{n}\left(t,-\frac{1}{2}\right)\right],
\ee
\be\label{re4}
\sigma_{2n+1}(t)=2\left[H_{n}\left(t,\frac{1}{2}\right)+H_{n+1}\left(t,-\frac{1}{2}\right)\right].
\ee
\noindent $\mathbf{Remark.}$ We could not derive the second order differential equation satisfied by $\sigma_{n}(t)$ directly from (\ref{p3}) and the above relations, since (\ref{p3}) is a non-linear differential equation.

From the definition of $R_{n}(t)$ and using the same method as above, we find
\be\label{re1}
R_{2n}(t)=2a_{n}\left(t,-\frac{1}{2}\right),
\ee
\be\label{re2}
R_{2n+1}(t)=2a_{n}\left(t,\frac{1}{2}\right),
\ee
for $n=0,1,2,\ldots$, where $a_{n}(t,\alpha)$ is defined in \cite{ChenIts} by
$$
a_{n}(t,\alpha):=\frac{t}{\tilde{h}_{n}(t,\alpha)}\int_{0}^{\infty}\frac{\tilde{P}_{n}^2(y,\alpha)}{y}y^{\alpha}\mathrm{e}^{-y-\frac{t}{y}}dy.
$$

By using Dyson's Coulomb fluid method, Chen and Chen \cite{Chen2015} derived the asymptotic expansions for the scaled $a_n(t),H_n(t,\alpha)$ and $\tilde{D}_{n}(t,\alpha)$ under the assumption that $n\rightarrow\infty$ and $t\rightarrow 0$ such that $s=(2n+1+\alpha)t$ or $s=(2n+1)t$ is fixed. Based on their results, we establish the asymptotics for our scaled $R_n(t),\sigma_n(t)$ and $D_n(t)$.

\begin{theorem}
Define
\[C_{1}(s):=\lim_{n\rightarrow\infty}\frac{R_{2n}\left(\frac{s}{2n+1}\right)}{\frac{s}{2n+1}},\qquad
C_{2}(s):=\lim_{n\rightarrow\infty}\frac{R_{2n+1}\left(\frac{s}{2n+1}\right)}{\frac{s}{2n+1}}.\]
 We have, for $s\rightarrow 0^{+}$,
$$
C_{1}(s)=-4+\frac{32s}{3}-\frac{256s^2}{15}+\frac{8192s^3}{315}-\frac{311296 s^4}{8505}+\frac{7733248 s^5}{155925}+O(s^6),
$$
$$
C_{2}(s)=4+\frac{32s}{3}+\frac{256s^2}{15}+\frac{8192s^3}{315}+\frac{311296 s^4}{8505}+\frac{7733248 s^5}{155925}+O(s^6),
$$
and for $s\rightarrow\infty$,
$$
C_{1}(s)=2s^{-\frac{1}{3}}+\frac{1}{3}s^{-\frac{2}{3}}+\frac{1}{108}s^{-\frac{4}{3}}-\frac{1}{648}s^{-\frac{5}{3}}+\frac{1}{324}s^{-2}-\frac{7}{5832 }s^{-\frac{7}{3}}+\frac{5}{1728 }s^{-\frac{8}{3}}+O(s^{-3}),
$$
$$
C_{2}(s)=2s^{-\frac{1}{3}}-\frac{1}{3}s^{-\frac{2}{3}}-\frac{1}{108}s^{-\frac{4}{3}}-\frac{1}{648}s^{-\frac{5}{3}}-\frac{1}{324}s^{-2}-\frac{7}{5832 }s^{-\frac{7}{3}}-\frac{5}{1728 }s^{-\frac{8}{3}}+O(s^{-3}).
$$
\end{theorem}
\begin{proof}
From (\ref{re1}) and (\ref{re2}), we have
$$
C_{1}(s)=2C\left(s,-\frac{1}{2}\right),\;\;\;\;C_{2}(s)=2C\left(s,\frac{1}{2}\right),
$$
where
$$
C(s,\alpha):=\lim_{n\rightarrow\infty}\frac{a_{n}\left(\frac{s}{2n+1+\alpha},\alpha\right)}{\frac{s}{2n+1+\alpha}}.
$$
It is given by (3.1) and (3.2) in \cite{Chen2015} that $C(s,\alpha)$ has the following expansions: for $s\rightarrow 0^{+}$,
\bea
C(s,\alpha)&=&\frac{1}{\alpha}-\frac{s}{\alpha^2(\alpha^{2}-1)}+\frac{3s^2}{\alpha^3(\alpha^2-1)(\alpha^2-4)}
-\frac{6(2\alpha^2-3)s^3}{\alpha^4(\alpha^2-1)^2(\alpha^2-4)(\alpha^2-9)}\nonumber\\
&+&\frac{5(11\alpha^2-36)s^4}{\alpha^5(\alpha^2-1)^2(\alpha^2-4)(\alpha^2-9)(\alpha^2-16)}\nonumber\\
&-&\frac{3(91\alpha^6-1115\alpha^4+4219\alpha^2-3600)s^5}{\alpha^6(\alpha^2-1)^3(\alpha^2-4)^2(\alpha^2-9)(\alpha^2-16)(\alpha^2-25)}+O(s^6),\nonumber
\eea
and for $s\rightarrow\infty$,
\bea
C(s,\alpha)&=&s^{-\frac{1}{3}}-\frac{\alpha}{3}s^{-\frac{2}{3}}+\frac{\alpha(\alpha^2-1)}{81}s^{-\frac{4}{3}}+\frac{\alpha^2(\alpha^2-1)}{243}s^{-\frac{5}{3}}
+\frac{\alpha(\alpha^2-1)}{243}s^{-2}\nonumber\\
&-&\frac{2\alpha^2(\alpha^2-1)(2\alpha^2-11)}{6561}s^{-\frac{7}{3}}-\frac{5\alpha(\alpha^2-1)(\alpha^4-\alpha^2-15)}{19683}s^{-\frac{8}{3}}+O(s^{-3}).\nonumber
\eea
According to these two expansions, we obtain the desired results.
\end{proof}
\begin{theorem}
Let
\[\sigma_{1}(s):=\lim_{n\rightarrow\infty}\sigma_{2n}\left(\frac{s}{2n+1}\right),\qquad
\sigma_{2}(s):=\lim_{n\rightarrow\infty}\sigma_{2n+1}\left(\frac{s}{2n+1}\right).\]
We have,
for $s\rightarrow 0^{+}$,
$$
\sigma_{1}(s)=\sigma_{2}(s)=-\frac{16 s^2}{3}-\frac{2048 s^4}{315}-\frac{3866624 s^6}{467775}+O(s^8),
$$
and for $s\rightarrow\infty$,
$$
\sigma_{1}(s)=\sigma_{2}(s)=-3 s^{\frac{2}{3}}-\frac{1}{18}-\frac{1}{432}s^{-\frac{2}{3}}-\frac{7}{7776}s^{-\frac{4}{3}}-\frac{31}{34992}s^{-2}
+O(s^{-\frac{8}{3}}).
$$
\end{theorem}
\begin{proof}
From (\ref{re3}) and (\ref{re4}) we find
$$
\sigma_{1}(s)=\sigma_{2}(s)=2\left[\mathcal{H}\left(s,\frac{1}{2}\right)+\mathcal{H}\left(s,-\frac{1}{2}\right)\right],
$$
where
$$
\mathcal{H}(s,\alpha):=\lim_{n\rightarrow\infty}H_{n}\left(\frac{s}{2n+1+\alpha},\alpha\right).
$$
It is given by (3.3) and (3.4) in \cite{Chen2015} that $\mathcal{H}(s,\alpha)$ has the following expansions: for $s\rightarrow 0^{+}$,
\bea
\mathcal{H}(s,\alpha)&=&-\frac{s}{2\alpha}+\frac{s^2}{4\alpha^2(\alpha^{2}-1)}-\frac{s^3}{2\alpha^3(\alpha^2-1)(\alpha^2-4)}
+\frac{3(2\alpha^2-3)s^4}{4\alpha^4(\alpha^2-1)^2(\alpha^2-4)(\alpha^2-9)}\nonumber\\
&+&\frac{(11\alpha^2-36)s^5}{2\alpha^5(\alpha^2-1)^2(\alpha^2-4)(\alpha^2-9)(\alpha^2-16)}\nonumber\\
&+&\frac{(91\alpha^{6}-1115\alpha^4+4219\alpha^2-3600)s^6}{4\alpha^6(\alpha^2-1)^3(\alpha^2-4)^2(\alpha^2-9)(\alpha^2-16)(\alpha^2-25)}+O(s^7),\nonumber
\eea
and for $s\rightarrow\infty$,
\bea
\mathcal{H}(s,\alpha)&=&-\frac{3}{4}s^{\frac{2}{3}}+\frac{\alpha}{2}s^{\frac{1}{3}}+\frac{1-6\alpha^2}{36}+\frac{\alpha(\alpha^2-1)}{54}s^{-\frac{1}{3}}
+\frac{\alpha^2(\alpha^2-1)}{324}s^{-\frac{2}{3}}+\frac{\alpha(\alpha^2-1)}{486}s^{-1}\nonumber\\
&-&\frac{\alpha^2(\alpha^2-1)(2\alpha^2-11)}{8748}s^{-\frac{4}{3}}-\frac{\alpha(\alpha^2-1)(\alpha^4-\alpha^2-15)}{13122}s^{-\frac{5}{3}}\nonumber\\
&-&\frac{\alpha^2(\alpha^2-1)(8\alpha^2-33)}{26244}s^{-2}+O(s^{-\frac{7}{3}}).\nonumber
\eea
From them, we readily get the desired expansions.
\end{proof}
\begin{theorem}
Write
\[\Delta_{1}(s):=\lim_{n\rightarrow\infty}\frac{D_{2n}(\frac{s}{2n+1})}{D_{2n}(0)},\qquad
\Delta_{2}(s):=\lim_{n\rightarrow\infty}\frac{D_{2n+1}(\frac{s}{2n+1})}{D_{2n+1}(0)}.\]
We have for $s\rightarrow 0^{+}$,
$$
\Delta_{1}(s)=\Delta_{2}(s)=\exp\left(-\frac{4 s^2}{3}-\frac{256 s^4}{315}-\frac{966656 s^6}{1403325}+O(s^8)\right),
$$
and for $s\rightarrow\infty$,
$$
\Delta_{1}(s)=\Delta_{2}(s)=\exp\left(\frac{\ln 2}{12}+3\zeta'(-1)-\frac{\ln s}{36}-\frac{9}{4} s^{\frac{2}{3}}+\frac{1}{576}s^{-\frac{2}{3}}+\frac{7}{20736}s^{-\frac{4}{3}}+O(s^{-2})\right),
$$
where  $\zeta(\cdot)$ is the Riemann zeta function.
\end{theorem}
\begin{proof}
It follows from (\ref{hd1}) and (\ref{hd2}) that
$$
\Delta_{1}(s)=\Delta_{2}(s)=\Delta\left(s,\frac{1}{2}\right)\Delta\left(s,-\frac{1}{2}\right),
$$
where
$$
\Delta(s,\alpha):=\lim_{n\rightarrow\infty}\frac{\tilde{D}_{n}\left(\frac{s}{2n+1+\alpha},\alpha\right)}{\tilde{D}_{n}(0,\alpha)}.
$$
It is given in Theorem 5 of \cite{Chen2015} that $\Delta(s,\alpha)$ has the following expansions: for $s\rightarrow 0^{+}$,
\bea
\Delta(s,\alpha)&=&\exp\Big(-\frac{s}{2\alpha}+\frac{s^2}{8\alpha^2(\alpha^{2}-1)}-\frac{s^3}{6\alpha^3(\alpha^2-1)(\alpha^2-4)}
+\frac{3(2\alpha^2-3)s^4}{16\alpha^4(\alpha^2-1)^2(\alpha^2-4)(\alpha^2-9)}\nonumber\\
&-&\frac{(11\alpha^2-36)s^5}{10\alpha^5(\alpha^2-1)^2(\alpha^2-4)(\alpha^2-9)(\alpha^2-16)}\nonumber\\
&+&\frac{(91\alpha^{6}-1115\alpha^4+4219\alpha^2-3600)s^6}{24\alpha^6(\alpha^2-1)^3(\alpha^2-4)^2(\alpha^2-9)(\alpha^2-16)(\alpha^2-25)}+O(s^7)\Big),\nonumber
\eea
and for $s\rightarrow\infty$,
\bea
\Delta(s,\alpha)&=&\exp\Big(c(\alpha)-\frac{9}{8}s^{\frac{2}{3}}+\frac{3\alpha}{2}s^{\frac{1}{3}}+\frac{1-6\alpha^2}{36}\ln s-\frac{\alpha(\alpha^2-1)}{18}s^{-\frac{1}{3}}-\frac{\alpha^2(\alpha^2-1)}{216}s^{-\frac{2}{3}}\nonumber\\
&-&\frac{\alpha(\alpha^2-1)}{486}s^{-1}+\frac{\alpha^2(\alpha^2-1)(2\alpha^2-11)}{11664}s^{-\frac{4}{3}}+\frac{\alpha(\alpha^2-1)(\alpha^4-\alpha^2-15)}{21870}s^{-\frac{5}{3}}\Big.\nonumber\\
&+&\Big.O(s^{-2})\Big),\nonumber
\eea
where $c(\alpha)$ is a constant independent of $s$, and is found to be
\be\label{cons}
c(\alpha)=\ln\frac{G(\alpha+1)}{(2\pi)^{\frac{\alpha}{2}}},
\ee
with $G(\cdot)$ denoting the Barnes G-function.

According to the above expansions, we have for $s\rightarrow 0^{+}$,
$$
\Delta_{1}(s)=\Delta_{2}(s)=\exp\left(-\frac{4 s^2}{3}-\frac{256 s^4}{315}-\frac{966656 s^6}{1403325}+O(s^8)\right),
$$
and for $s\rightarrow\infty$,
$$
\Delta_{1}(s)=\Delta_{2}(s)=\exp\left(c\left(\frac{1}{2}\right)+c\left(-\frac{1}{2}\right)-\frac{\ln s}{36}-\frac{9}{4} s^{\frac{2}{3}}+\frac{1}{576}s^{-\frac{2}{3}}+\frac{7}{20736}s^{-\frac{4}{3}}+O(s^{-2})\right).
$$From (\ref{cons}), we find
$$
c\left(\frac{1}{2}\right)+c\left(-\frac{1}{2}\right)=\ln \left(G\left(\frac{1}{2}\right)G\left(\frac{3}{2}\right)\right)=\ln \left(G^2\left(\frac{1}{2}\right)\Gamma\left(\frac{1}{2}\right)\right).
$$
The equation (6.39) in Voros \cite{Voros} shows that
$$
G\left(\frac{1}{2}\right)=2^{\frac{1}{24}}\pi^{-\frac{1}{4}}\mathrm{e}^{\frac{3\zeta'(-1)}{2}},
$$
where $\zeta(\cdot)$ is the Riemann zeta function. It follows that
$$
c\left(\frac{1}{2}\right)+c\left(-\frac{1}{2}\right)=\frac{\ln 2}{12}+3\zeta'(-1).
$$
This completes the proof.
\end{proof}

\noindent $\mathbf{Remark.}$ The above constant $\frac{\ln 2}{12}+3\zeta'(-1)$ is called Dyson's constant, which was conjected by Dyson \cite{Dyson} in the asymptotic expansion of a Fredholm determinant with the sine kernel, and was found by Widom \cite{Widom} in the study on the asymptotics of the Toeplitz determinants. It was rigorously proved by Krasovsky \cite{Krasovsky} using the Riemann-Hilbert method, and by Ehrhardt \cite{Ehrhardt} using a different approach. Lyu, Chen and Fan \cite{Lyu2018} also obtained this constant in the study of the gap probability in the Gaussian unitary ensemble; see Theorem 4.1 therein.

\section{Conclusion}
We study the Hankel determinant $D_{n}(t)$ generated by a singularly perturbed Gaussian weight in this paper.
By using the ladder operator approach, we obtain three auxiliary quantities related to the Hankel determinant, $R_{n}(t)$, $r_{n}(t)$ and $\sigma_{n}(t)$. We show that each of them satisfies a non-linear second order difference equation and a non-linear second order differential equation, from which the Painlev\'{e} III$'$ appears. We also consider the large $n$ asymptotics of $R_{n}(t)$, $\sigma_{n}(t)$ and $D_{n}(t)$ under a double scaling $s=(2n+1)t,\;n\rightarrow\infty, t\rightarrow 0$ such that $s$ is fixed. The asymptotic expansions of the scaled $R_{n}(t)$, $\sigma_{n}(t)$ and $D_{n}(t)$ for large $s$ and small $s$ are obtained.

\section*{Acknowledgments}
Chao Min was supported by the Scientific Research Funds of Huaqiao University under grant number 600005-Z17Y0054.
Yang Chen was supported by the Macau Science and Technology Development Fund under grant numbers FDCT 130/2014/A3, FDCT 023/2017/A1 and by the University of Macau under grant numbers MYRG 2014-00011-FST, MYRG 2014-00004-FST.

\end{document}